\documentclass[11pt,a4paper]{article} 
\usepackage[a4paper,hmargin=2.8cm,vmargin=3cm]{geometry}
\usepackage{amsmath,amsthm,amsfonts,amssymb}
\usepackage{authblk}
   
\usepackage{graphicx}    
 \usepackage{todonotes}
\usepackage[bookmarksnumbered=true]{hyperref} 
\usepackage[capitalise]{cleveref}
\crefname{equation}{}{}
\Crefname{lem}{Lemma}{Lemmas}   
\usepackage{tikz-cd}
\usetikzlibrary{positioning,calc,fadings,decorations.pathreplacing,shadings}

\newtheorem{thm}{Theorem}[section]

\newtheorem{prp}[thm]{Proposition}
\newtheorem{lem}[thm]{Lemma}

\theoremstyle{remark}

\renewcommand{\theequation}{\thesection.\arabic{equation}}

\numberwithin{equation}{section}

\newcommand{\fock}{\mathcal{F}}		
\newcommand{\di}{{\textnormal{d}}}		

\newcommand{\Lcal}{\mathcal{L}}		
\newcommand{\Mcal}{\mathcal{M}}

\newcommand{\Ucal}{\mathcal{U}}

\newcommand{\Tbb}{\mathbb{T}}

\newcommand{\Ncal}{\mathcal{N}}		
\newcommand{\Hcal}{\mathcal{H}}		

\newcommand{\hc}{\textnormal{h.c.}}		

\newcommand{\Rbb}{\mathbb{R}}		
\newcommand{\Cbb}{\mathbb{C}}		
\newcommand{\Nbb}{\mathbb{N}}		
\newcommand{\Zbb}{\mathbb{Z}}

\renewcommand{\Re}{\operatorname{Re}} 	
\renewcommand{\Im}{\operatorname{Im}} 	

\newcommand{\norm}[1]{\lVert#1\rVert}	

\newcommand{\tr}{\operatorname{tr}}

\newcommand{\sgn}{\operatorname{sgn}}
\newcommand{\tagg}[1]{ \stepcounter{equation} \tag{\theequation} \label{eq:#1} } 

\newcommand{\F}{\textnormal{F}} 

\title{Two Comments on the Derivation of the Time--Dependent Hartree--Fock Equation} 
\author[1,a]{Niels Benedikter}
\author[1,b]{Davide Desio}
\affil[1]{Universit\`a degli Studi di Milano, Dipartimento di Matematica, Via Cesare Saldini 50, 20133 Milano, Italy}
\affil[a]{ORCID: \href{https://orcid.org/}{0000--0002--1071--6091}, e--mail: \href{mailto:niels.benedikter@unimi.it}{niels.benedikter@unimi.it}}
\affil[b]{ORCID: \href{https://orcid.org/}{0000--0001--9840--3809}, e--mail: \href{mailto:davide.desio@studenti.unimi.it}{davide.desio@studenti.unimi.it}}

\begin{document}
\maketitle

\begin{abstract}
We revisit the derivation of the time--dependent Hartree--Fock equation for interacting fermions in a regime coupling a mean--field and a semiclassical scaling, contributing two comments to the result obtained in 2014 by Benedikter, Porta, and Schlein. First, the derivation holds in arbitrary space dimension. Second, by using an explicit formula for the unitary implementation of particle--hole transformations, we cast the proof in a form similar to the coherent state method of Rodnianski and Schlein for bosons.
\end{abstract}

\tableofcontents

\section{Interacting Fermi Gases at High Density}
In condensed matter, one--, two--, and three--dimensional quantum systems are realized. In a basic approximation, an ordinary piece of metal can be modelled as a gas of interacting fermions in three dimensions; transistor--like semiconductor structures can in first approximation be considered as a two--dimensional electron gas; and the one--dimensional electron gas may be used as a simplified model of a carbon nanotube. Mathematically even these simple models are difficult to study because a quantum system of $N$ particles is described by a vector in the antisymmetrized tensor product of $N$ copies of $L^2(\Rbb^{d})$. As $N$ is easily of the order of $10^{4}$ and more likely up to $10^{23}$, numerical methods quickly find their limits in the analysis of the many--body Schrödinger equation. One way of overcoming this difficulty is the use of effective  equations: in idealized physical regimes the Schrödinger equation may be approximated by equations involving fewer degrees of freedom. For fermions, Hartree--Fock theory is such an approximation: one considers initial data given as an antisymmetrized elementary tensor (a Slater determinant) and then projects \cite{Lub08,BSS18} the many--body Schrödinger evolution on the submanifold of antisymmetrized elementary tensors. In the present note we show that the quantitative error estimates obtained in \cite{BPS14b} for the Hartree--Fock equation apply to all space dimensions, and we reformulate the proof using an explicit formula for the unitary implementation of a particle--hole transformation, thus casting it in a form completely analogous to the analysis of the bosonic mean--field limit by the coherent state method of \cite{RS09}. 

\medskip

In the following paragraphs we will introduce the many--body Schrödinger equation, the scaling regime, reduced density matrices, and the Hartree--Fock equation.

\paragraph{Fundamental Description: The Schrödinger Equation} The fundamental description is given by the Hamiltonian (with a coupling constant $\lambda \in \Rbb$)
\begin{equation} \label{eq:hamiltonian}
 H_N := -\sum_{i=1}^N   \Delta_{i} +   \lambda \sum_{1\leq i<j\leq N} V(x_i - x_j)\;,      
\end{equation}
a self--adjoint operator on the antisymmetric subspace $L^2_\textnormal{a}(\Rbb^{dN})$ of $L^2(\Rbb^d)^{\otimes N} \simeq L^2(\Rbb^{dN})$, i.\,e., functions $\psi \in L^2(\Rbb^{dN})$ satisfying
\begin{equation} \label{eq:antisymmetry}
\psi(x_1,x_2,\ldots,x_N) = \sgn(\sigma) \psi(x_{\sigma(1)},x_{\sigma(2)},\ldots,x_{\sigma(N)}) \quad \textnormal{for } \sigma \in \mathcal{S}_N\;.
 \end{equation}
The evolution of  initial data $\psi_0 \in L^2_\textnormal{a}(\Rbb^{dN})$ is given by the Schrödinger equation
\begin{equation}\label{eq:schroedinger}
 i \partial_t \psi_t = H_N \psi_t\;.
\end{equation}
Our goal is to approximate solutions of \cref{eq:schroedinger} by the time--dependent Hartree--Fock equation. Considering an appropriate scaling of the system parameters with the particle number $N$, one can prove estimates on the difference asymptotically as $N \to \infty$. In the next paragraph we discuss our choice of such a scaling regime.

\paragraph{Coupled Mean--Field and Semiclassical Scaling Regime}
No approximation applies to all physical situations. The situation we consider was introduced by \cite{NS81,Spo81} for deriving the Vlasov equation from quantum mechanics. In this setting the density of the system is large but the interaction between any pair of particles weak, so that mean--field like behaviour may be expected. To derive the precise choice of parameters we consider for the moment the torus $\Tbb^d := \Rbb^d/2\pi\Zbb^d$ instead of $\Rbb^d$. The simplest fermionic wave functions are antisymmetrized elementary tensors (i.\,e., Slater determinants) 
\begin{equation} \label{eq:psiF}
\psi(x_1,x_2,\ldots x_N) = f_1 \wedge \cdots \wedge f_N(x_1,\ldots,x_N) = (N!)^{-1/2} \det\big(f_j(x_i)\big)_{i,j=1,\ldots,N}\;.
\end{equation}
Ignoring for the moment the interaction $V$, the ground state is the Slater determinant of $N$ plane waves $f_j(x) := (2\pi)^{-d/2} e^{i k_j \cdot x}$  where 
\[k_j \in B_\F := \{ k \in \Zbb^d : \lvert k\rvert \leq k_\F \} \;.\]
If instead of using $N$ as independent parameter we use the Fermi momentum $k_\F >0$, i.\,e., define $N := \lvert B_\F \rvert$ as a function of $k_\F$, then the Slater determinant of the plane waves with $k_j \in B_\F$ is the unique minimizer of the non--interacting Hamiltonian. Since $k_\F \sim N^{1/d}$, the total kinetic energy becomes
\begin{equation} \label{eq:BF}
\langle\psi,\bigg( -\sum_{i=1}^N   \Delta_{i} \bigg)\psi\rangle=\sum_{k\in B_\F} |k|^2 \sim  N^{1 + \frac{2}{d}} \qquad \textnormal{as }k_\F \to \infty\;.
\end{equation}
Now let us bring back the interaction into the game, and consider its expectation value in the same Slater determinant of plane waves. To have a large--$N$ limit in which neither kinetic nor interaction energy (as a sum over pairs being of order $\lambda N^2$) dominates, we set
\[
\lambda := N^{\frac{2}{d}-1}\;.
\]
The particles most affected by the interaction are those close to the surface of the Fermi ball $B_\F$, i.\,e., with momenta $\lvert k\rvert \sim k_\F \sim N^{1/d}$. Like their momentum, also their velocity is of order $N^{1/d}$. Therefore we study times of order $N^{-1/d}$; the accordingly rescaled equation is
\[
 i N^{1/d} \partial_t \psi_t = \Bigg(\sum_{i=1}^N - \Delta_{i} + N^{\frac{2}{d} - 1} \sum_{1\leq i<j\leq N} V(x_i - x_j) \Bigg) \psi_t \;.
\]
Introducing an effective Planck constant
\[\hbar := N^{-1/d}\]
and multiplying by $\hbar^2$, we obtain the Schrödinger equation we study in this note:
\begin{equation}\label{eq:coreeq}
 i\hbar \partial_t \psi_t =  \Bigg(\sum_{i=1}^N - \hbar^2 \Delta_{i} + \frac{1}{N} \sum_{1\leq i<j\leq N}V(x_i - x_j) \Bigg) \psi_t\;.
\end{equation}
Other scaling limits, with weaker interaction or shorter time scale, have been considered in \cite{BGGM03,BGGM04,FK11,PP16,BBP+16}.

\paragraph{Reduced Density Matrices}
Given an $N$--particle observable $A$, i.\,e., a self--adjoint operator on $L^2_\textnormal{a}(\Rbb^{dN})$, its expectation value in a state $\psi \in L^2_\textnormal{a}(\Rbb^{dN})$ can be written with a trace over $L^2_\textnormal{a}(\Rbb^{dN})$ in Dirac's bra--ket notation as 
\[
\langle\psi,A\psi\rangle = \tr_N \Big( \lvert \psi\rangle \langle \psi \rvert A \Big)\;.
\]
Simpler observables are the averages of one--particle observables: if $a$ is an operator on $L^2(\Rbb^d)$ and $a_j$ means $a$ acting on the $j$--th of $N$ tensor factors, $a_j := 1 \otimes \cdots \otimes 1 \otimes a \otimes 1 \otimes \cdots \otimes 1$, the expectation value can be written with a partial trace over $N-1$ tensor factors as
\[\frac{1}{N} \sum_{j=1}^N \langle \psi, a_j \psi\rangle = \langle \psi, a_1 \psi\rangle = \tr_1 \big( a \tr_{N-1} \lvert \psi\rangle\langle \psi\rvert \big)\;.\]
The one--particle reduced density matrix, an operator on the one--particle space $L^2(\Rbb^d)$, is
\begin{equation}\label{eq:1pdm}
\gamma^{(1)}_\psi := N \tr_{N-1} \lvert \psi\rangle\langle \psi\rvert \;.\end{equation}
As a trace class operator, the spectral theorem permits to decompose it as
\[\gamma^{(1)}_\psi = \sum_{j \in \Nbb} \lambda_j \lvert \varphi_j\rangle \langle \varphi_{j} \rvert\;, \qquad \varphi_j \in L^2(\Rbb^d)\;, \qquad \lambda_{j} \in \Rbb\;.\]
In particular we may speak of its integral kernel and its ``diagonal'' (representing the density of particles in position space), defined by
\[\gamma^{(1)}_\psi(x;y) := \sum_{j \in \Nbb}\lambda_{j}\varphi_{j}(x)  \overline{\varphi_{j}(y)}\;, 
\qquad \gamma^{(1)}_\psi(x;x) := \sum_{j \in \Nbb} \lambda_j \lvert \varphi_{j}(x)\rvert^2 \;.\]
A Slater determinant $\psi(x_1,x_2,\ldots x_N) = (N!)^{-1/2} \det(\varphi_j(x_i))$ is an example of a quasi--free state, and as such uniquely (up to a phase factor) determined by its one--particle reduced density matrix. The one--particle reduced density matrix of a Slater determinant is a rank--$N$ projection, i.\,e., of the $\lambda_j$ in the spectral decomposition $N$ have value $1$ and the rest are $0$.

\paragraph{Effective Description: Hartree--Fock Theory}
In Hartree--Fock theory, attention is restricted to Slater determinants, with the choice of the orbitals $\varphi_j$ to be optimized. Projecting the time--dependent Schr\"odinger equation locally onto the tangent space of this submanifold (i.\,e., applying the Dirac--Frenkel principle, see \cite{Lub08,BSS18}) one obtains the time--dependent Hartree--Fock equations (a system of $N$ non--linear coupled equations)
\begin{equation}\label{eq:hforb}
i \hbar \partial_t \varphi_{j,t} = - \hbar^2 \Delta \varphi_{j,t} + \frac{1}{N}\sum_{i=1}^N \Big(V \ast \lvert \varphi_{i,t}\rvert^2\Big) \varphi_{j,t} - \frac{1}{N} \sum_{i=1}^N\Big(V \ast (\varphi_{j,t} \overline{\varphi_{i,t}}\big)\Big) \varphi_{i,t} \;.
\end{equation}
In terms of the one--particle density matrix $\omega_{N,t} := \sum_{j=1}^N \lvert \varphi_{j,t}\rangle \langle \varphi_{j,t} \rvert$ they take the form
\begin{align*}
 & i\hbar \partial_t \omega_{N,t} = [-\hbar^2\Delta + (V \ast \rho_t) -X_t, \omega_{N,t}]\;. \tagg{hf}
\end{align*}
The term $V \ast \rho_t$ with $\rho_t(x) := \omega_{N,t}(x;x)$ is a multiplication operator called the direct term. The exchange term $X_t$ is defined by its integral kernel $X_t(x;x') = V(x - x') \omega_{N,t}(x;x')$.

Given a rank--$N$ projection operator $\omega_N$ as initial data, the solution of \cref{eq:hf} is for all times a rank--$N$ projection operator. From its spectral decomposition, fixing the phase ambiguity appropriately, one obtains the $N$ orbitals  solving \cref{eq:hforb}.

\section{Main Result}
Let $X$ be the one--particle position operator on $L^2(\Rbb^d)$, i.\,e., the multiplication operator $X\psi(x) = x\psi(x)$ for $x \in \Rbb^d$. Let $P := -i \hbar \nabla$ be the one--particle momentum operator.
We have now introduced everything necessary to state our main result:
\begin{thm}[Validity of the Hartree--Fock Equation]\label{thm:hf}
Let $d \in \Nbb$. Consider an interaction potential $V \in L^1(\Rbb^d)$ with Fourier transform satisfying $\int \di p (1+\lvert p\rvert)^2 \lvert \hat{V}(p)\rvert < \infty$. Let $\omega_N$ be a sequence of rank--$N$ projection operators on $L^2(\Rbb^d)$, and assume there exist $C_X >0$ and $C_P >0$ such that for all $i \in \Nbb \cap [1,d]$ and for all $N \in \mathbb{N}$ we have
\begin{equation} \label{eq:semiclass}
 \sup_{\alpha \in \Rbb^d} \frac{\lVert [e^{i\alpha\cdot X},\omega_N] \rVert_\textnormal{tr}}{1+\lvert \alpha\rvert} \leq N\hbar\,C_X \;, \qquad \lVert [P,\omega_N] \rVert_\textnormal{tr} \leq N \hbar\,C_P\;.
\end{equation}
(The latter estimate is to be read in $\ell^2$--sense with respect to the components of the momentum operator, i.\,e., $\lVert [P,\omega_N] \rVert_\textnormal{tr} = (\sum_{i=1}^d \lVert [P_i,\omega_N] \rVert_\textnormal{tr}^2)^{1/2}$.)
Let $\psi_{N,0}$ be the Slater determinant uniquely (up to a phase factor) determined by $\omega_N$. Let $\gamma^{(1)}_{N,t}$ be the one--particle reduced density matrix of the solution $\psi_{N,t} := e^{-i H_N t/\hbar} \psi_{N,0}$ of the Schrödinger equation. Let $\omega_{N,t}$ be the solution of the Hartree--Fock equation \cref{eq:hf} with initial data $\omega_N$. Let $q_0:=\int \di p  (1+|p|)^2 \lvert \hat{V}(p)\rvert$, then for all $t \in \Rbb$ and for $N \in \Nbb$ sufficiently large we have
\begin{equation}\label{eq:hfest}
 \norm{ \gamma^{(1)}_{N,t} - \omega_{N,t}}_{\operatorname{tr}} \leq \sqrt{N} 6\exp\Big(2^3 \frac{C_X + C_P}{\max\{2,q_0\}} e^{2\max\{2,q_0\} |t|}\Big) \;.
\end{equation}
\end{thm}
The trace norm estimate of order $N^{1/2}$ is to be compared to the  triangle inequality which would yield $2N$. As in \cite{BPS14b}, the result may be generalized to $k$--particle reduced density matrices; and as in \cite{BPS14a} it can be generalized to relativistic massive particles.

The assumption \cref{eq:semiclass} is realized by the Fermi ball (see \cref{eq:psiF,eq:BF}), which however is stationary under the Hartree--Fock evolution (for $\hat{V} \geq 0$ it is even the global minimizer \cite[Theorem~A.1]{BNP+21b}). The assumption is also realized by some examples with non--trivial Hartree--Fock evolution such as the ground state of non--interacting fermions in a harmonic trap \cite{Ben22} or even a general trapping potential \cite{FM20}. Actually, in \cite{Ben22} a bound was shown for $\lVert [X_i,\omega_N] \rVert_\textnormal{tr}$ instead of $\sup_{\alpha\in \Rbb^d}\lVert [e^{i\alpha\cdot X},\omega_N] \rVert_\textnormal{tr}(1+\lvert \alpha\rvert)^{-1}$. These are related by
\[
[\omega_{N}, e^{i\alpha\cdot X}] = e^{i\alpha\cdot X} \int_0^1 \di \lambda\, \frac{\di}{\di \lambda} \left(e^{-i\alpha\cdot X\lambda} \omega_{N} e^{i\alpha\cdot X\lambda} \right) = e^{i\alpha\cdot X} \int_0^1 \di \lambda\, e^{-i\alpha\cdot X\lambda} [\omega_{N},i\alpha\cdot X] e^{i\alpha\cdot X\lambda}\;,
\]
so (as shown similarly also in \cite[Corollary~1.3]{FM20})
\begin{align*}
\sup_{\alpha\in \Rbb^d} \frac{\tr \lvert [\omega_{N},e^{i\alpha\cdot X}]\rvert}{1+\lvert \alpha\rvert}   & \leq \sup_{\alpha\in \Rbb^d} \frac{1}{1+\lvert \alpha\rvert} \tr \lvert [\omega_{N},\alpha\cdot X]\rvert \leq \sup_{\alpha\in \Rbb^d} \frac{1}{1+\lvert \alpha\rvert}\sum_{j=1}^d \lvert \alpha_j\rvert \tr \lvert [\omega_{N},X_j]\rvert\\
& \leq \sup_{\alpha\in \Rbb^d} \frac{\lvert \alpha\rvert}{1+\lvert \alpha\rvert} \Bigg[\sum_{j=1}^d \big(\tr \lvert [\omega_{N},X_j]\rvert\big)^2\Bigg]^{1/2} = \lVert [\omega_{N},X] \rVert_\textnormal{tr}\;.
\end{align*}

\medskip

Singular interaction potentials $V$ were considered in \cite{PRSS17,Saf18} for initial data which is stationary under the time--dependent Hartree--Fock equation.
The Hartree--Fock equation has also been derived for initial data given by a mixed state \cite{BJP+16}. This has been generalized to singular interaction potentials, including the Coulomb potential and the gravitational attraction in \cite{CLS21,CLS22}. The validity of the Hartree--Fock equation has been derived for extended Fermi gases in three dimensions by \cite{FPS22}.
Next--order corrections (the random phase approximation) and a Fock space norm approximation, however only for approximately bosonic excitations of the stationary Fermi ball, have been obtained in \cite{BNP+22}, based on the bosonization method developed in \cite{BNP+20,BNP+21b,BPSS21,Ben21}. For a further discussion of different levels of dynamical approximation, see the review \cite{Ben22}.

\section{Proof}
Let us quickly fix some notation. Fermionic Fock space is defined as
\[\fock := \Cbb \oplus \bigoplus_{n=1}^\infty L^2_{\textnormal{a}}(\Rbb^{dn})\;.\]
For $f,g \in L^2(\Rbb^d)$, the well--known creation and annihilation operators $a^*(f)$ and $a(g)$ satisfy the canonical anticommutator relations
\[\{a(f),a^*(g)\} = \langle f;g\rangle\;, \qquad \{a(f),a(g)\} = 0 = \{a^*(f),a^*(g)\} \;.\]
In the fermionic case these operators satisfy for all $\psi \in \fock$ the bounds
\[\norm{a(f)\psi}_\fock \leq \norm{f}_{L^2(\Rbb^d)} \norm{\psi}_\fock\;, \qquad \norm{a^*(f)\psi}_\fock \leq \norm{f}_{L^2(\Rbb^d)} \norm{\psi}_\fock\;.\]
The particle number operator is denoted by $\Ncal$. The vacuum is $\Omega = (1,0,0,0,\ldots)$, the (up to a phase) unique vector in the null space of all annihilation operators. This implies $\Ncal \Omega =0$. Moreover, given any operator $A$ on $L^2(\Rbb^d)$ with integral kernel $A(x;y)$, its second quantization written using the operator valued distributions associated to the creation and annihilation operators is
\[\di\Gamma(A) := \int \di x \di y A(x;y) a^*_x a_y \;.\]
The following lemma collects standard bounds; see \cite[Section~3]{BPS14b} for their proof.
\begin{lem}[Bounds for Second Quantization]\label{lem:2ndq}
Let $\psi \in \fock$ and let $A$ be an operator on $L^2(\Rbb^d)$. Then we have
\begin{align}
 \norm{\di\Gamma(A) \psi}_\fock & \leq \norm{A}_\textnormal{op} \norm{\Ncal\psi}_\fock \;, \\
 \norm{\di\Gamma(A) \psi}_\fock & \leq \norm{A}_\textnormal{HS} \norm{\Ncal^{1/2}\psi}_\fock \;, \\
 \norm{\di\Gamma(A) \psi}_\fock & \leq \norm{A}_\textnormal{tr} \norm{\psi}_\fock \;.
\end{align}
Moreover, if $A$ has an integral kernel $A(x;y)$, then
\begin{align}
 \norm{\int \di x\di y A(x;y) a_x a_y \psi}_\fock & \leq \norm{A}_\textnormal{HS} \norm{\Ncal^{1/2}\psi}_\fock\;, \\
 \norm{\int \di x\di y A(x;y) a^*_x a^*_y \psi}_\fock & \leq 2 \norm{A}_\textnormal{HS} \norm{ (\Ncal + 1)^{1/2} \psi}_\fock \;,
 \intertext{and}
 \norm{\int \di x\di y A(x;y) a_x a_y \psi}_\fock & \leq 2\norm{A}_\textnormal{tr} \norm{\psi}_\fock\;, \\
 \norm{\int \di x\di y A(x;y) a^*_x a^*_y \psi}_\fock & \leq 2\norm{A}_\textnormal{tr} \norm{\psi}_\fock\;.
 \end{align}
\end{lem}
Finally, note that the definition of the one--particle reduced density matrix may be generalized to $\psi \in \fock$ by setting
\begin{align} \label{eq:1pdm2ndq}
 \gamma^{(1)}_\psi(x;y) := \langle \psi, a^*_y a_x \psi\rangle\;.
\end{align}
In fact, if $\psi \in L^2_\textnormal{a}(\Rbb^{dN})$ is considered as a subspace of Fock space, then this $\gamma^{(1)}_\psi$ is exactly the integral kernel of the operator defined in \cref{eq:1pdm}.

\subsection{Implementation of Particle--Hole Transformations}
Let $(\varphi_j)_{j=1}^N$ be an orthonormal system in $L^2(\Rbb^d)$. The main difference in the present proof with respect to \cite{BPS14b} is the use of the following definition:
\begin{align} \label{eq:implementation}
R_{N} := \prod_{j=1}^N \left( a^*(\varphi_j) + a(\varphi_j) \right)\;.
\end{align}
This is a unitary map on Fock space which maps the vacuum on a Slater determinant,
\[R_N \Omega = \prod_{j=1}^N a^*(\varphi_j) \Omega = (N!)^{-1/2} \det(\varphi_j(x_i)) \;,\]
and satisfies
\begin{equation}R_N a^*(\varphi_j) R_N^* = \left\{ \begin{array}{rl} (-1)^{N+1} a(\varphi_j) & \textnormal{for } j \leq N \\ (-1)^N a^*(\varphi_j) & \textnormal{for } j > N\;. \end{array} \right.\end{equation}
The formula \cref{eq:implementation} is an implementation of a particle--hole transformation as constructed by abstract Bogoliubov theory in \cite{BPS14b}. We got aware of this formula from \cite[Eq.~(57)]{Lil22}.

Moreover it is convenient to introduce the operators
\begin{equation}Q_{N} := \sum_{j=1}^N \lvert \varphi_{j}\rangle \langle \overline{\varphi_j} \rvert\;, \qquad P_{N} := 1 - \sum_{j=1}^N \lvert \varphi_{j}\rangle \langle \varphi_j \rvert \;, \label{eq:PQ0}\end{equation}
where $\overline{\varphi_j}$ is the complex conjugation of $\varphi_j \in L^2(\Rbb^d)$. The action of the particle--hole transformation on the creation and annihilation operators can then be computed to be
\begin{align}\label{eq:PQ}
 R^*_N a_x R_N = (-1)^N \left( a(P_{N,x}) - a^*(Q_{N,x}) \right)\,, \quad R^*_N a^*_x R_N = (-1)^N \left( a^*(P_{N,x}) - a(Q_{N,x}) \right) \,, 
\end{align}
where $Q_N(x;y)$ and $P_N(x;y)$ are (formal) integral kernels of the operators $Q_N$ and $P_N$, and $Q_{N,x}(y) := Q_N(y;x)$, $P_{N,x}(y) := P_N(y;x)$ for all $y\in \Rbb^d$.

We are going to use \cref{eq:implementation} to construct a unitary fluctuation dynamics as in \cite{RS09}. The proof of the main theorem will then be obtained by an application of the Grönwall lemma, following the strategy of \cite{BPS14b}.

\subsection{Many--Body Analysis}
The Hamiltonian $H_N$ may be represented on Fock space as
\[\Hcal_N := \hbar^2 \int \di x \nabla_x a^*_x \nabla_x a_x + \frac{1}{2N} \int \di x \di y\, V(x-y) a^*_x a^*_y a_y a_x \;.\]
In fact, considering $L^2_\textnormal{a}(\Rbb^{dN})$ as a subspace of $\fock$, we have $\Hcal_N\restriction_{L^2_\textnormal{a}(\Rbb^{dN})} = H_N$. Since we consider only initial data in the $N$--particle subspace and the evolution preserves particle numbers (i.\,e., $[\Ncal,\Hcal_N] = 0$) we can use $\Hcal_N$ in the place of $H_N$.

Let $\omega_{N,t}$ be the solution of the time--dependent Hartree--Fock equation \cref{eq:hf} (for a discussion of the well--posedness see, e.\,g., \cite{BSS18}) with initial data $\omega_N$. Let $\varphi_{j,t}$, with $j=1,2\ldots N$ be the corresponding orthonormal systems of orbitals, and $R_{N,t}$ the correspondingly constructed particle--hole transformation as in \cref{eq:implementation}. We define the unitary  fluctuation dynamics
\begin{equation}\label{eq:fluctdyn}
\Ucal_N(t,s) := R^*_{N,t} e^{-i(t-s)\Hcal_N/\hbar} R_{N,s}\;.
\end{equation}
The advantage of introducing the fluctuation dynamics $\Ucal_N$ is the following representation of the difference that we want to estimate:
\begin{lem}[Trace Norm Difference]\label{lem:tracenormdiff}
Let $\omega_N$ be a rank--$N$ projection operator, and let $\omega_{N,t}$ be its evolution under the time--dependent Hartree--Fock equation \cref{eq:hf}. Let $R_{N,0}$ and $R_{N,t}$ be the corresponding particle--hole transformations. Let moreover $\psi_{N,0} := R_{N,0} \Omega$ be a Slater determinant and $\psi_{N,t} := e^{-i\Hcal_n t/\hbar} \psi_{N,0}$ its many--body Schrödinger evolution. Let $\gamma^{(1)}_{N,t}$ be the one--particle reduced density matrix associated to $\psi_{N,t}$. Then for all $t\in \Rbb$ we have
 \[ \norm{\gamma^{(1)}_{N,t} -\omega_{N,t} }_\textnormal{tr} \leq \left(2 + 4\sqrt{N}\right) \langle \Ucal_N(t,0) \Omega, (\Ncal+1) \Ucal_N(t,0) \Omega \rangle \;.\]
\end{lem}
The proof of \cref{lem:tracenormdiff} is unchanged from \cite[Section~4]{BPS14b}.

\medskip

As in \cref{eq:PQ0}, we introduce also for the Hartree--Fock evolved orbitals $\varphi_{j,t}$ the operators
\[Q_{N,t} := \sum_{j=1}^N \lvert \varphi_{j,t}\rangle \langle \overline{\varphi_{j,t}} \rvert\;, \qquad P_{N,t} := 1 - \sum_{j=1}^N \lvert \varphi_{j,t}\rangle \langle \varphi_{j,t} \rvert \;.\]

The novelty of the present note lies in the use of the explicit formula \cref{eq:implementation} for computing the time derivative of $\langle \Ucal_N(t,0) \Omega, (\Ncal+1) \Ucal_N(t,0) \Omega \rangle$. The computation is then essentially identical to that given for bosons in the derivation of the Hartree equation by the coherent states method of \cite{RS09}, simply with the Weyl operators $W(\sqrt{N}\varphi_t)$ replaced by $R_{N,t}$. The result of the computation constitutes the following proposition. 
\begin{prp}[Generator of Fluctuations]\label{prp:generator}
Given $\Ucal_N(t;s)$ by \cref{eq:fluctdyn}, we define the generator of fluctuations $\Lcal_N(t)$ by
\[i\hbar \partial_t \Ucal_N(t;s) = \Lcal_N(t) \Ucal_N(t;s)\;.\]
Then we have
\begin{equation}\label{eq:generator}
\Lcal_N(t) = \mathcal{A}_N(t) + \mathcal{B}_N(t) + \mathcal{C}_N(t) + \Mcal_N(t)+ \hc\;,  
\end{equation}
where
\begin{align*}
 \mathcal{A}_N(t) & := \frac{1}{2N} \int \di x \di y V(x-y) a^*(P_{N,t,x}) a^*(P_{N,t,y}) a^*(Q_{N,t,y}) a^*(Q_{N,t,x})\\
 \mathcal{B}_N(t) & := \frac{1}{N} \int \di x \di y V(x-y) a^*(P_{N,t,x}) a^*(P_{N,t,y}) a^*(Q_{N,t,x}) a(P_{N,t,y})\\
 \mathcal{C}_N(t) & :=  \frac{1}{N} \int \di x \di y V(x-y) a^*(P_{N,t,x}) a^*(Q_{N,t,x}) a^*(Q_{N,t,y}) a(Q_{N,t,y})
\end{align*}
and the operator $\Mcal_N(t)$ commutes with the particle number operator: $[\Mcal_N(t),\Ncal] = 0$ for all $N \in \Nbb$ and all $t \in \Rbb$.
\end{prp}
\begin{proof}
In this proof $\Mcal_N(t)$ denotes an operator commuting with the number of particles operator, potentially changing from line to line without further comment. Obviously
 \[\Lcal_N(t) =  (i \hbar\partial_t R^*_{N,t}) R_{N,t} + R^*_{N,t} \Hcal_N(t) R_{N,t} \;.\]
 The contribution of $R^*_{N,t}\Hcal_N R_{N,t}$ is easily computed using \cref{eq:PQ}, expanding all the products and using the canonical anticommutator relations to obtain an expression completely in normal order (i.\,e., with creation operators to the left of annihilation operators). One finds
 \begin{align*}
  R^*_{N,t} \hbar^2 \int \di x \nabla_x a^*_x \nabla_x a_x R_{N,t} & = \sum_{j=1}^Na^*(\hbar^2\Delta\varphi_{j,t})a^*(\varphi_{j,t})-\sum_{k,j=1}^{N}\langle \varphi_{j,t},\hbar^2\Delta\varphi_{k,t} \rangle
  a^*(\varphi_{j,t})a^*(\varphi_{k,t})\\
  & \quad + \hc + \Mcal_N(t) \tagg{kinetic}
 \end{align*}
 and
 \begin{align*}
 &R^*_{N,t} \frac{1}{2N} \int \di x \di y V(x-y) a^*_x a^*_y a_y a_x R_{N,t} \\
 & = \frac{1}{2N} \int d{x}d{y}\,V({x}-{y}) 
\Big[a^*(P_{t,{x}})a^*(P_{t,{y}})a^*(Q_{t,{y}})a^*(Q_{t,{x}}) + 2a^*(P_{t,{x}})a^*(P_{t,{y}})a^*(Q_{t,{x}})a(P_{t,{y}}) \\
& \hspace{10.5em} + 2a^*(P_{t,{x}})a^*(Q_{t,{x}})a^*(Q_{t,{y}})a(Q_{t,{y}})-2 \langle Q_{t,{y}},Q_{t,{y}} \rangle a^*(P_{t,{x}})a^*(Q_{t,{x}})\\
&\hspace{10.5em} +2\langle Q_{t,{y}},Q_{t,{x}} \rangle a^*(P_{t,{x}})a^*(Q_{t,{y}})\Big]+ \hc + \Mcal_N(t)\;.   \tagg{interaction}
 \end{align*}
The summand involving the time derivative is slightly more complicated to compute. We define $R(h):=a^*(h)+a(h)$ for $h \in L^2(\Rbb^d)$ and observe that $\{R(\varphi_{l,t}),R(\varphi_{k,t})\}=2\delta_{l,k}$. Thus
\begin{align*}
& (i\hbar \partial_tR^*_{N,t})R_{N,t} \\
& =i\hbar R(\partial_t\phi_{N,t})R(\phi_{N,t})+
i\hbar\sum_{j=1}^{N-1}\prod_{k=0}^{j-1}R(\varphi_{N-k,t})R(\partial_t\varphi_{N-j,t})\prod_{m=N-j}^NR(\varphi_{m,t})\\
&=\sum_{k=1}^{N}i\hbar R(\partial_t\varphi_{k,t})R(\varphi_{k,t})-2i\hbar\sum_{k=1}^{N}\sum_{j=1}^{k-1}\Re\langle \varphi_{k,t},\partial_t\varphi_{k-j,t}\rangle R(\varphi_{k,t})R(\varphi_{k-j,t})\\
&=\sum_{k=1}^{N}i\hbar R(\partial_t\varphi_{k,t})R(\varphi_{k,t})-i\hbar\sum_{k=1}^{N}\sum_{\substack{j=1\\j\ne k}}^N\langle \varphi_{k,t},\partial_t\varphi_{j,t}\rangle R(\varphi_{k,t})R(\varphi_{j,t})\\
&=\sum_{k=1}^{N}a^*(i\hslash\partial_t\varphi_{k,t})a^*(\varphi_{k,t})-\sum_{k=1}^{N}\sum_{j=1}^{N}\langle\varphi_{k,t},i\hslash\partial_t\varphi_{j,t}\rangle a^*(\varphi_{k,t})a^*(\varphi_{j,t})+ \hc +\Mcal_N(t) \;.
\end{align*}
 Thus
 \begin{equation}
 \begin{split}\label{eq:tderiv}
  (i \hbar\partial_t R^*_{N,t}) R_{N,t}&=\sum_{k=1}^{N}a^*(i\hslash\partial_t\varphi_{k,t})a^*(\varphi_{k,t})-\sum_{k=1}^{N}\sum_{j=1}^{N}\langle\varphi_{k,t},i\hslash\partial_t\varphi_{j,t}\rangle a^*(\varphi_{k,t})a^*(\varphi_{j,t})\\
  & \quad + \hc + \Mcal_N(t).
 \end{split}
 \end{equation}
Summing \cref{eq:kinetic}, \cref{eq:interaction}, and \cref{eq:tderiv},  the Hartree--Fock equation \cref{eq:hf} implies the cancellation of all the quadratic (containing products of two creation or annihilation operators) terms that do not commute with $\Ncal$. The remaining terms are as claimed in \cref{eq:generator}. 
\end{proof}

Using \cref{prp:generator} for the generator of fluctuations, one easily proves the following lemma, where we are back at \cite[Proposition~3.3]{BPS14b}.
\begin{lem}
With $\mathcal{A}_N(t)$, $\mathcal{B}_N(t)$, and $\mathcal{C}_N(t)$ as defined in the previous proposition we have
\begin{align*}
 & i\hbar \frac{\di}{\di t} \langle \Ucal_N(t,0) \Omega, (\Ncal + 1) \Ucal_N(t,0) \Omega \rangle = \langle \Ucal_N(t,0) \Omega, [\Ncal,\Lcal_N(t)] \Ucal_N(t,0) \Omega \rangle \\
 & = -2i\Im \langle \Ucal_N(t,0) \Omega, \left(4 \mathcal{A}_N(t) + 2 \mathcal{B}_N(t) + 2 \mathcal{C}_N(t) \right) \Ucal_N(t,0) \Omega \rangle \;. \tagg{tobeest}
\end{align*}
\end{lem}
One now writes $V(x-y)$ in \cref{eq:tobeest} in terms of its Fourier transform and then, using \cref{lem:propagation} and \cref{lem:2ndq} one shows as in \cite[Lemma~3.5]{BPS14b} that 
\[
\begin{split}
& \left\lvert \hbar \frac{\di}{\di t} \langle \Ucal_N(t,0) \Omega, (\Ncal + 1) \Ucal_N(t,0) \Omega \rangle\right\rvert \\
& \leq \hbar\,2^4(C_X + C_P) e^{2\max\{2,q_0\}|t|}\langle \Ucal_N(t,0) \Omega, (\Ncal + 1) \Ucal_N(t,0) \Omega \rangle
\end{split}
\]
for all $t \in \Rbb$, whence the main result follows by Grönwall's lemma. \qed

\subsection{Propagation of Commutator Bounds}
The following lemma propagates the bounds on the commutators from the initial data to all times. Though stated in \cite[Proposition~3.4]{BPS14b} only for $d=3$, the proof is without modifications valid for any $d \in \Nbb$. This lemma refers only to the Hartree--Fock evolution.
\begin{lem}[Propagation of Commutator Bounds]\label{lem:propagation}
Consider an interaction potential $V \in L^1(\Rbb^d)$ with Fourier transform satisfying $\int \di p (1+\lvert p\rvert)^2 \lvert \hat{V}(p)\rvert < \infty$. Let $\omega_N$ be a sequence of rank--$N$ projection operators on $L^2(\Rbb^d)$, and assume there exist $C_X >0$ and $C_P >0$ such that for all $i \in \Nbb \cap [1,d]$ and for all $N \in \mathbb{N}$ we have
\[
 \sup_{\alpha \in \Rbb^d} \frac{\lVert [e^{i\alpha\cdot X},\omega_N] \rVert_\textnormal{tr}}{1+\lvert \alpha\rvert} \leq N \hbar\, C_X\;, \qquad \lVert [P_i,\omega_N] \rVert_\textnormal{tr} \leq N \hbar\, C_P\;.
\]
Let $\omega_{N,t}$ be the solution of the Hartree--Fock equation \cref{eq:hf} with initial data $\omega_N$. Let $q_0:=\int \di p  (1+|p|)^2 \lvert \hat{V}(p)\rvert$. Then for all $t \in \Rbb$ and all $N \in \Nbb$ we have
\[ 
\sup_{\alpha \in \Rbb^d} \frac{\lVert [e^{i\alpha\cdot X},\omega_N] \rVert_\textnormal{tr}}{1+\lvert \alpha\rvert} \leq N \hbar (C_X + C_P) e^{2\max\{2,q_0\} \lvert t\rvert}, \ \lVert [P_i,\omega_{N,t}] \rVert_\textnormal{tr}  \leq N \hbar (C_X + C_P) e^{2\max\{2, q_0\} \lvert t\rvert}\,.
\]
\end{lem}
The exponential time dependence may not be optimal; however, for our proof the important aspect of these bounds is that we gain at all times a factor $\hbar$ with respect to the naive bound $\lVert [e^{i\alpha\cdot X},\omega_N] \rVert_\textnormal{tr} \leq \norm{e^{i\alpha\cdot X}\omega_N}_\textnormal{tr} + \norm{\omega_Ne^{i\alpha\cdot X}}_\textnormal{tr} = 2 \norm{\omega_N}_\textnormal{tr} = 2N$.

\section*{Acknowledgements and Declarations}
NB has been supported by Gruppo Nazionale per la Fisica Matematica (GNFM)  in Italy and the European Research Council (ERC) through the Starting Grant \textsc{FermiMath}, grant agreement nr.~101040991. The authors acknowledge the support of Istituto Nazionale di Alta Matematica ``F.~Severi'', through the Intensive Period ``INdAM Quantum Meetings (IQM22)''. The authors do not have any conflicts of interest to disclose.

\bibliographystyle{alpha}
\bibliography{semiclassicalstructure}{}   

\end{document}